\documentclass[12pt]{article}
\usepackage{amsmath,amssymb,amsthm,fullpage,setspace,hyperref,nicefrac,paralist,booktabs,comment,setspace,natbib,bm,soul}

\hypersetup{colorlinks=true, linkcolor=blue,citecolor=blue}

\setcitestyle{authoryear,open={(},close={)}}

\usepackage[small]{titlesec}
\usepackage{caption}
\newtheorem{theorem}{Theorem}

\newtheorem{corollary}{Corollary}

\newtheorem{proposition}{Proposition}

\newcommand{\ee}{\mathrm{e}}

\newcommand{\RR}{\mathbb{R}}
\newcommand{\PP}{\mathbb{P}}
\newcommand{\EE}{\mathbb{E}}

\renewcommand{\P}[1]{{\PP\left[ #1 \right]}}
\newcommand{\E}[1]{{\EE\left[ #1 \right]}}
\def\dd{\mathrm{d}}

\title{{\bf \Large Background Risk and Small-Stakes Risk Aversion}\thanks{We want to thank Ned Augenblick, Nick Barberis, Sebastian Ebert, Armin Falk, Stefano DellaVigna, Paul Heidhues, Matthew Rabin, Todd Sarver, Dmitry Taubinsky, Charlie Sprenger and Juuso V\"alim\"aki for helpful comments and discussions.}}

\author{ \large Xiaosheng Mu\thanks{Princeton University. Email: xmu@princeton.edu. Xiaosheng Mu acknowledges the hospitality of Columbia University and the Cowles Foundation at Yale University, which hosted him during parts of this research.} \ \ \ Luciano Pomatto\thanks{Caltech. Email: 	luciano@caltech.edu.} \ \ \ %
Philipp Strack\thanks{Yale University. Email:  philipp.strack@yale.edu. Philipp Strack was supported by a Sloan fellowship.} \ \ \ %
Omer Tamuz\thanks{Caltech. Email: tamuz@caltech.edu. Omer Tamuz was supported by a grant from the Simons Foundation (\#419427), by a BSF grant (\#2018397) and by a Sloan fellowship.}}

\date{\today}

\begin{document}

\maketitle

\begin{abstract}



We show that under plausible levels of background risk, no theory of choice under risk---such as expected utility theory, prospect theory, or rank dependent utility---can simultaneously satisfy the following three economic postulates:
(i) Decision makers are risk-averse over small gambles,
(ii) they respect stochastic dominance, and
(iii) they account for background risk.
%

\end{abstract}

%

\section{Introduction}

 How humans evaluate the trade-off between risks and rewards is one of the core questions in economics. In solving this trade-off many people exhibit \emph{small-stakes risk aversion}: small, actuarially favorable gambles---such as a lottery where one loses $\$10$ or wins $\$11$ with equal probability---are often rejected.
  In this paper we study how background risk (stemming, for instance, from investments in the stock market or health conditions) affects risk attitudes towards small gambles. We argue that for plausible levels of background risk, no preference can explain small-stakes risk aversion without violating monotonicity with respect to first-order stochastic dominance.
 
%

 For example, we show that the aforementioned gamble, where one loses $\$10$ or wins $\$11$ with equal probability, must be accepted by any decision maker with monotone preferences, provided that she faces an independent background risk with exponential tails and standard deviation as small as $\$200$. 
 This insight is not specific to this particular example. Our main result, Theorem~\ref{thm:main}, establishes that for every bounded gamble with strictly positive mean there exists a level of background risk such that accepting the gamble is dominant with respect to first-order stochastic dominance. 
 We furthermore provide a simple bound on how large the background risk needs to be for this conclusion to hold for a given gamble.
 This bound implies that under plausible levels of background risk it is stochastically dominant to accept many gambles that are frequently rejected.
 
 \medskip
 
 To illustrate the result, consider a decision maker who is facing a binary gamble $X$ under which she gains $G$ dollars or loses $L$ dollars with probability 1/2. If the gamble was taken in isolation, then the choice of whether or not to accept it would depend on her preferences between $X$ and a sure outcome of $0$. But if the decision maker is facing an independent background risk $W$ regarding her wealth, then the relevant choice is between $W$, if the gamble is rejected, and $W + X$, if the gamble is accepted.
%
\bgroup
\def\arraystretch{1.4}
\begin{table}[t]
    \centering
    \begin{tabular}{l|l l l}
        \toprule
         \textbf{Gamble} & \multicolumn{3}{c}{\textbf{StDeviation of Background Risk: $\sigma$}} \\
         Gain/Loss & Laplace & Logistic  
         & Normal\textsuperscript{*} \\
         \midrule
         $\$11/\$10$ & $\sigma \geq \$156$ & $\sigma \geq \$200$ & $\sigma \geq \$3319$\\
         $\$55/\$50$ & $\sigma \geq \$779$ & $\sigma \geq \$999$ & $\sigma \geq \$7422$\\
         $\$110/\$100$ & $\sigma \geq \$1557$ & $\sigma \geq \$1997$ & $\sigma \geq \$10{,}498$\\
         $\$550/\$500$ & $\sigma \geq \$7785$ & $\sigma \geq \$9984$ & $\sigma \geq \$23{,}526$\\
         $\$1100/\$1000$ & $\sigma \geq \$15{,}569$ & $\sigma \geq \$19{,}967$ & $\sigma \geq \$33{,}361$\\
         \bottomrule
    \end{tabular}
    \caption{Standard deviation of background risk sufficient for every decision maker with monotone preferences to accept various fifty-fifty gambles under different distributional assumptions on the background risk. The numbers displayed are bounds derived from our Theorem \ref{thm:main} and Corollary \ref{cor:normal}. \\
    \textsuperscript{*}The normal distribution has mean $\$100{,}000$, and in this case the decision maker's wealth is bounded below by \$0, according to a limited liability assumption which we discuss in \S\ref{sec:model}.}
    \label{tab:main-table}
\end{table}
\egroup
 Table~\ref{tab:main-table} displays different levels of standard deviations and distributional assumptions for the background risk under which $W + X$ dominates $W$ in first-order stochastic dominance. For example, in the case of the fifty-fifty gamble with gain $G=11$ and loss $L=10$, it is dominant to accept  whenever the decision maker's wealth has a standard deviation higher than $\$200$ and is either Laplace or Logistically distributed. 


 The standard deviation of many real-life risks plausibly exceeds this threshold by a large margin. For example, an investor who has $\$100{,}000$ in an S\&P 500 index fund reasonably faces a wealth risk with standard deviation $\$1{,}000$ (or 1\%) for the value of her portfolio at the end of each \emph{day}, and $\$15{,}000$ at the end of the year.\footnote{See, e.g., \cite*{bardgett2019inferring}.} Nevertheless, small-stakes gambles are commonly rejected. \cite*{barberis2006individual} find that among clients of a U.S.\ bank with median wealth exceeding \$10 million, the rejection rate of a hypothetical $\$550/\$500$ gamble is $71\%$. 
 Table~\ref{tab:main-table} suggests that this behavior is inconsistent with these investors taking into account even the short-term background risk they face.

 Table~\ref{tab:main-table} applies to all monotone preferences. Since monotonicity is, by itself, a very weak assumption, we need some conditions on the distribution of the background risk $W$ to ensure that $W+X$ stochastically dominates $W$. In particular, $W$ must have full support and the left tail of the distribution must be sufficiently thick. However, we also show that if the decision maker is protected by limited liability, so that her final wealth cannot go below a certain threshold, then no assumptions on the tails of the background risk are required. Moreover, for particular preference specifications, our results continue to hold under smaller background risks. For example, in \S\ref{appx:CPT} of the Appendix, we calculate the level of background risk needed for a decision maker with Cumulative Prospect Theory preferences to accept various small gambles, under the same parameter values calibrated by \cite{tversky1992advances}. Compared to Table~\ref{tab:main-table}, the required levels of standard deviation are significantly smaller: a fifty-fifty gamble with gain $G = 11$ and loss $L = 10$ must be accepted whenever the background risk has a standard deviation higher than \$62, and is either Laplace, Logistic or Normally distributed.


 \medskip

 Our analysis shows a tension between three natural requirements for any theory of choice under risk: (i) risk aversion over small gambles, (ii) monotonicity with respect to first-order stochastic dominance, and (iii) choices are made accounting for background risk.
 As (i) is commonly observed in real world choices, and relaxing (ii) is widely considered unappealing, our results suggest that theories that do not account for narrow framing---whereby independent sources of risk are evaluated separately by the decision maker---cannot explain commonly observed choices among risky alternatives.


%
%
%
 

\paragraph{Related Literature.}
 \cite{arrow1970essays} and \cite{pratt1964risk} show that under expected utility and a twice-differentiable utility function, a decision maker accepts any actuarially favorable gamble, provided that it is scaled to be small enough. 
 \cite{rabin2000calibration} shows that the degree of concavity necessary for expected utility theory to explain small-stakes risk aversion leads to implausible choices over large lotteries.\footnote{For example, given any risk-averse expected utility preference, if a gamble where one loses $\$100$ or wins $\$110$ with equal probability is rejected at all wealth levels below $\$300{,}000$, then a gamble where one loses $\$2000$ and wins $\$12{,}000{,}000$ with equal probability must also be rejected at wealth levels below $\$290{,}000$. See also \cite{Hansson1988calibration} for an early example illustrating this point. \cite{zambrano2020risk} provides further results on calibrating expected utility preferences.}  
 The literature has then suggested two different ways of explaining small-stakes risk aversion: (i) Considering more general preferences that allow for \emph{loss aversion} or \emph{first-order risk aversion} and (ii) relaxing the assumption that small gambles are evaluated in context by allowing for \emph{narrow framing}.\footnote{For example \cite{RabinThaler2001anomaly} page 226 writes: ``If expected utility theory doesn’t explain the modest-scale risk aversion we observe, what does? We think that the right explanation incorporates two concepts that have been mentioned before in the ``Anomalies'' series: loss aversion and mental accounting.'' Narrow framing is also emphasized in \cite{CoxSadiraj2006calibration}, \cite{Rubinstein2006dilemma} and \cite{harrison2018integration} among other papers.}
 
 A variety of alternatives to expected utility theory feature first-order risk aversion \citep{SegalSpivak1990firstorder,AngBekaertLiu2005stocks}, i.e.\ non-vanishing risk aversion over small risks.
 Among the most prominent ones are prospect theory \citep{KahnemanTversky1979prospecttheory}, rank dependent utility \citep{Quiggin1982rdu}, disappointment aversion \citep{Gul1991disappoint} and reference dependent preferences \citep{koszegi2007reference}.\footnote{In the context of stochastic choice, \cite*{KhawEtAl2020cognitive} provide a new model of first-order risk aversion based on cognitively imprecise representation of the decision environment.} While these theories can explain small-stakes risk aversion when such risks are evaluated in isolation, subsequent work in the literature suggests that they may be unable to do so once background risk is taken into account.
 
 In a dynamic context, \cite*{barberis2006individual} show that in the presence of large background risk, Rabin's critique extends to recursive disappointment-averse preferences.\footnote{More recently, \cite{sarver2018dynamic} proposes a new class of dynamic mixture-averse preferences that can reconcile risk attitudes toward both small and large gambles, provided the amount of background risk is moderate.} \cite{koszegi2007reference} illustrate through an example that a decision maker with reference dependent preference will for some gambles ``approach risk neutrality [...] even for relatively limited amounts of background risk.''\footnote{\cite{gollier2004economics} surveys the earlier literature on background risk within expected utility.} More generally, \cite{SafraSegal2008calibration} show that for any risk-averse preference admitting a G\^{a}teaux differentiable representation, a decision maker who rejects a small gamble with positive mean under all background risks must also reject highly favorable large gambles. The literature has thus established that for different classes of risk-averse preferences, small-stakes risk aversion under all background risks leads to the rejection of large attractive gambles. 
 
 Our analysis strengthens this conclusion in three main ways: 
 First, unlike the existing literature, we do not focus on a specific class of preferences, but instead consider general monotone preferences, regardless of whether they are risk-averse, risk-loving or neither. Our analysis therefore applies to rank-dependent utility, disappointment averse preferences, reference dependent preferences, as well as many other common models. 
 Second, our results do not require assumptions on the decision maker's behavior with regard to large, and possibly hypothetical, gambles.

 Finally, our results are not limited to showing the existence of a distribution of background risk under which accepting a given small gamble is dominant. Instead, we show that this conclusion holds whenever the background risk is sufficiently large, in a sense that we can exactly quantify. Thus, whenever a decision maker is observed to reject small gambles despite facing considerable background risks, an analyst can conclude that she either violates stochastic dominance or performs narrow framing (so that small gambles are not fully combined with background risks). Our stronger result thus facilitates experimental tests of narrow framing. 
 
 
 \cite{tarsney2018exceeding} and \cite*{pomatto2018stochastic} independently show that background risk can induce stochastic dominance between gambles that are not already ranked. 
 In particular, \cite*{pomatto2018stochastic} prove that if $X$ and $Y$ are random variables with $\E{X} > \E{Y}$, then there exists an independent random variable $Z$ such that $X+Z$ first-order stochastically dominates $Y+Z$. In their result, the random variable $Z$ is tailored to the particular $X$ and $Y$ involved. In contrast, in this paper we show that all sufficiently large heavy-tailed background risks ensure stochastic dominance, and explicitly characterize the necessary size of the background risk.
 
 Our analysis emphasizes the importance of narrow framing for theories of choice under risk. The idea that narrow framing affects decision makers' risk attitudes goes back to \cite{BenartziThaler1995myopic}, who propose an explanation of the equity premium puzzle via loss aversion and narrow framing. Later work by \cite{GneezyPotters1997experiment} and \cite{ThalerEtAl1997experiment} experimentally tests the hypothesis that decision makers who account for more background risk tend to be more tolerant of independent risky gambles. To the best of our knowledge, our paper is the first to demonstrate theoretically that whenever a plausible level of background risk is taken into consideration, minimal assumptions on preferences lead to risk-neutral behavior for small gambles. Interestingly, the connection between narrow framing and stochastically dominant choices is also explored in \cite{Rabin2009dominance}, who show that any decision maker who engages in narrow framing and has non-CARA preferences will make choices that, when combined, violate stochastic dominance.

\section{Model}\label{sec:model}

 A decision maker faces a choice between accepting or rejecting a gamble described by a bounded random variable $X$ that takes negative values with positive probability.\footnote{We thus rule out the trivial case where $X \geq 0$ and the gamble almost surely pays out a positive amount.} The decision maker's wealth $W$ is random and independent of $X$, and accepting the gamble leads to final wealth $W + X$. We interpret $W$ as background risk the decision maker faces when considering whether or not to accept the gamble. We assume $W$ is distributed according to a density $g \colon \RR \to \RR_{+}$ that has full support, is eventually decreasing, and is piece-wise continuously differentiable.\footnote{Formally, $g$ is continuous, and there exists a positive integer $n$ and numbers $-\infty = a_0 < a_1 < \dots < a_{n-1} < a_n = \infty$ such that $g$ is continuously differentiable on each of the open intervals $(a_{i-1}, a_i)$. Moreover, we require for each $i$, the limits $\lim_{x \nearrow a_{i}} g'(x)$ and $\lim_{x \searrow a_{i}} g'(x)$ exist and are finite. Using the continuity of $g$ and the Mean Value Theorem, it can be shown that these limits are the left and right derivatives of $g$ at $a_i$, respectively.} This is a weak technical assumption that holds for many common distributions, like the Normal, Logistic, or Laplace distributions.
 
 \paragraph{Monotone Preferences.}

 When $W + X$ dominates $W$ with respect to first-order stochastic dominance, we say that accepting $X$ is \textit{dominant}. We make no assumptions on the decision maker's behavior except that she accepts every dominant gamble.

 
 Recall that a random variable $Y$ \textit{first-order stochastically dominates} another random variable $Z$ if for every $a \in \RR$ it holds that $\P{Y \geq a} \geq \P{Z \geq a}$. Within the expected utility framework, a preference respects first-order stochastic dominance if and only if it is represented by an increasing utility function. More generally, a preference is monotone with respect to stochastic dominance if and only if it satisfies two conditions: (i) the preference between any two random variables depends only on their distributions, and (ii) $Y$ is preferred to $Z$ whenever $Y \geq Z$ almost surely.\footnote{This equivalent formulation is based on the well-known characterization that if $Y$ first-order stochastically dominates $Z$, then there exist two other random variables $\tilde{Y},\tilde{Z}$ with the same distributions as $Y$ and $Z$, respectively, and such that $\tilde{Y} \geq \tilde{Z}$ almost surely.} Thus, the assumption that behavior is consistent with first-order stochastic dominance expresses the idea that the decision maker's choice between accepting or rejecting the gamble $X$ at wealth $W$ depends only on the distributions of $W + X$ and $W$, and that more money is preferred over less.
 
 As discussed in the introduction, consistency with respect to stochastic dominance is a weak assumption satisfied by virtually all preference specifications studied in decision theory and behavioral economics. In fact, the assumption is often satisfied even when choice behavior is not described by means of a single complete and transitive preference relation over wealth distributions, as in models of random expected utility. 
 
 \paragraph{Limited Liability.} In many contexts it is natural to assume that the decision maker is protected by limited liability, in the sense that her wealth cannot go below a  bound $\ell \in \mathbb{R}$. Given a random variable $Z$, we denote by $(Z)_\ell = \max\{Z,\ell\}$ the variable truncated at $\ell$. Under limited liability, comparison between risky prospects boils down to a comparison between their truncated counterparts: The decision maker receives the amount $(W+X)_\ell$ when she accepts the gamble, and $W_\ell$ when she rejects it. Accordingly, we say that accepting $X$ is \textit{dominant} if $(W + X)_\ell$ dominates $(W)_\ell$ in first-order stochastic dominance. Thus, we assume the decision maker prefers more wealth than less, and prefers a lower probability of reaching the liability bound. Aside from this, we make no assumptions on what reaching the bound implies for the decision maker.

\section{Main Results}
 
 Before stating the formal results we introduce two indices for quantifying the magnitude of the background risk and the riskiness of a gamble.
 
\paragraph{Size of the Background Risk} Given a background risk $W$ with  
density $g$, we define its \textit{exponential size} $\mathrm{S}(W)$ as 
\[
    \mathrm{S}(W) = \left( \sup_{a} \frac{g'(a)}{g(a)} \right)^{-1},
\]
 where the supremum is taken over points $a$ where $g$ is differentiable. 
 We say that $W$ is \textit{heavy left-tailed} if $\mathrm{S}(W)>0$. This restriction excludes distributions with thin tails such as the Normal distribution, but includes common parametric families such as Logistic or Laplace.
 
 Intuitively, the exponential size is a measure of how likely large losses are relative to small ones. The larger $\mathrm{S}(W)$, the more slowly the density $g$ increases, and hence the thicker is the left tail of the distribution. More formally, the density $g(a)$ associated with a loss $a<0$ must lie above the exponential function $g(0) \cdot \ee^{\frac{a}{\mathrm{S}(W)}}$ and thus cannot vanish faster than exponentially. Similar to the standard deviation, the exponential size is positive homogeneous and independent of the location of $W$ (i.e.\ it satisfies $\mathrm{S}(tW) = t\mathrm{S}(W)$ for $t > 0$ and $\mathrm{S}(W) = \mathrm{S}(W+c)$ for any $c \in \mathbb{R}$).\footnote{The exponential size is however not symmetric: in general $\mathrm{S}(-W) \neq \mathrm{S}(W)$. This is reflective of the fact that the exponential size mainly concerns the left tail of the background risk distribution.}

 

 \paragraph{Riskiness of the Gamble} We quantify the riskiness of a gamble in terms of the Aumann-Serrano index. Given a gamble $X$ with positive expectation and positive probability of being negative, \cite{aumann2008economic} define its riskiness $\mathrm{R}(X)$ as the reciprocal of the level of absolute risk aversion at which a decision maker with CARA expected utility preferences is indifferent between accepting and rejecting $X$. Formally, $\mathrm{R}(X)$ is defined as the (unique) positive real number solving the equation
 \[
     \EE\left[\ee^{-\frac{1}{\mathrm{R}(X)}X}\right] = 1 \,.
 \]
 Intuitively, a gamble that is assigned a higher index is riskier because it is accepted by a smaller pool of risk-averse CARA decision makers. 
 
 \cite{aumann2008economic} provide an axiomatic foundation for the index. The riskiness index can also be related to simpler concepts such as the expectation of a gamble and the size of its support. In particular, whenever a gamble $X$  is supported in an interval $[-M,M]$ and has expectation $\epsilon > 0$, its riskiness index satisfies $\mathrm{R}(X) \leq \frac{M^2}{\epsilon}$, as we show in Proposition~\ref{prop:risk-bound} in the Appendix.
 
 \medskip
 
 The next theorem is the main technical result of the paper. It shows that it is dominant to accept any gamble that has positive expectation and whose riskiness is bounded by the exponential size of the background risk. The complete proof is provided in \S\ref{sec:main-proof} below. 
 
\begin{theorem}
\label{thm:main}
Under any given background risk $W$, it is dominant to accept every gamble $X$ with positive expectation and riskiness $\mathrm{R}(X) \leq \mathrm{S}(W)$.
\end{theorem}

 
This result implies that accepting a gamble $X$ is dominant provided it is not excessively risky compared to the size of the background risk $W$. An important corollary is the following: 
 
\begin{corollary}\label{cor:scaled-gamble}
Suppose the background risk $W$ is heavy left-tailed, i.e.\ $\mathrm{S}(W) > 0$, or the decision maker is protected by limited liability for some liability bound $\ell$. Then, for every gamble $X$ with positive expectation, accepting the scaled gamble $tX$ is dominant for all $t>0$ small enough.
\end{corollary}

In the case of limited liability, the next result---another corollary of Theorem~\ref{thm:main}---shows that a large enough Normal background risk suffices to make any actuarially favorable gamble dominant.

\begin{corollary}\label{cor:normal}
Under a background risk distributed Normally with mean $\mu$ and standard deviation $\sigma$, it is dominant for a decision maker with limited liability bound $\ell$ to accept every gamble $X$ with positive expectation, riskiness $\mathrm{R}(X)$ and maximum $\max[X]$ if 
$$
  \sigma \geq \sqrt{\mathrm{R}(X)(\mu-\ell+\max[X])}.
$$
\end{corollary}
As we show in the Appendix, similar lower bounds can be derived for other distributions, and are especially simple to calculate for log-concave densities such as the Normal density. 

\bigskip
 
Theorem~\ref{thm:main} and its corollaries allow us to provide quantitative estimates for the level of background risk under which accepting a gamble is dominant. To illustrate, consider a gamble $X$ that pays $\$110$ and $-\$100$ with equal probability and a background risk $W$ having Laplace distribution, i.e.\ following a density $g(a) = \frac{1}{2\lambda} \ee^{-\vert a \vert/\lambda}$. The parameter $\lambda$ coincides with the exponential size of $W$, and the distribution has standard deviation $\sigma = \sqrt{2}\lambda$. It follows from Theorem~\ref{thm:main} that it is dominant to accept the gamble $X$ as long as the standard deviation of $W$ satisfies
\[
    \sigma \geq \sqrt{2} \mathrm{R}(X) \approx \$1557.
\]
If instead $W$ follows a Logistic distribution with standard deviation $\sigma$, then its exponential size is $\mathrm{S}(W) = \frac{\sqrt{3}}{\pi}\sigma$. Thus accepting the gamble is dominant provided
 \[
    \sigma \geq \frac{\pi}{\sqrt{3}} \mathrm{R}(X) \approx \$1997.
 \]
For a decision maker with limited liability bound $\ell=0$ and Normally distributed $W$ with mean $\$100{,}000$, it is dominant to accept $X$ if
\[
    \sigma \geq \sqrt{\mathrm{R}(X) \cdot 100{,}110} \approx \$10{,}498.
\]
Table~\ref{tab:main-table} is constructed from similar calculations.

 \section{Proof of Theorem~\ref{thm:main}}
 \label{sec:main-proof}
Let $s = \mathrm{S}(W)$, and let $G$ denote the c.d.f.\ of $W$. The result is vacuous when $s = 0$, so we focus below on $s > 0$. A gamble $X$ has positive expectation and Aumann-Serrano index $\mathrm{R}(X) \leq s$ if and only if it is accepted by a decision maker with constant absolute risk aversion of $1/s$, i.e.\ if and only if
\begin{equation}\label{eq:CARA-conditions}
    \E{\ee^{-\frac{1}{s} X}} \leq 1.
\end{equation}
On the other hand, accepting the gamble $X$ is dominant if
\[
    \P{W + X \leq a}\leq \P{W \leq a} = G(a) \, \text{~for all~} a\in\RR \,.
\]
Since the gamble $X$ and the background risk $W$ are independent, we have that
\[
    \P{W + X \leq a} = \P{W \leq a - X} = \E{ G(a-X) }. 
\]
Thus, accepting $X$ is dominant if and only if
\begin{equation}\label{eq:eu-representation}
    \E{ G(a-X) } \leq G(a) \, \text{~for all~} a\in\RR \,.
\end{equation}

Inequality \eqref{eq:eu-representation} can be interpreted as saying that a decision maker with expected utility preferences and utility function $G$ rejects the gamble $-X$ at every wealth level $a$. This analogy is useful as it allows us to understand stochastic dominance through the behavior of a hypothetical expected utility decision maker whose utility coincides with the distribution $G$ of the background risk. The preference defined by $G$ is in general neither globally risk-averse nor risk-loving. Since the expectation of $-X$ is negative, equation \eqref{eq:eu-representation} is equivalent to imposing that the utility function $G$ is ``not too risk-loving.'' Below we formalize this intuition.

 
As shown by \eqref{eq:CARA-conditions} and \eqref{eq:eu-representation} above, to say that it is dominant to accept every gamble $X$ with positive mean and riskiness $\mathrm{R}(X) \leq s$ is equivalent to saying that for every gamble $X$,
\begin{equation}\label{eq:diffidence}
    \E{\ee^{-\frac{1}{s} X}} \leq 1 \implies \E{ G(a - X) } \leq G(a) \, \text{~for all~} a\in\RR.
\end{equation}
That is, any gamble $-X$ that is rejected by a decision maker with (risk-loving) CARA utility $U(a) = \ee^{\frac{a}{s}}$ is also rejected by a decision maker with utility $G$, at all wealth levels. In other terms,
we require $G$ to be globally more risk-averse than $U$ in the sense of Arrow-Pratt. 

To show this, let us first suppose for simplicity that the density $g$ is everywhere differentiable. In this case, since the exponential size of $W$ is $s$, we have $g'(a)/g(a) \leq 1/s$ for every $a$. Thus, the Arrow-Pratt index of risk aversion is everywhere higher for $G$ than for $U$:
\[
    -\frac{G''(a)}{G'(a)} = -\frac{g'(a)}{g(a)} \geq -\frac{1}{s} = - \frac{U''(a)}{U'(a)}  \, \text{~for all~} a \in \RR\,.
\]
It follows that $G$ is indeed more risk-averse than $U$.

For the general case where $g$ is only piece-wise continuously differentiable, note that $U(a) = \ee^{\frac{a}{s}}$ is a strictly increasing function, enabling us to write $G(a) = \phi\left(U(a)\right) = \phi(\ee^{\frac{a}{s}})$ for some increasing function $\phi$ defined on $(0, \infty)$. Our goal is to show that $\phi$ is concave, so that $G$ is more risk-averse than $U$. Note that 
\[
g(a) = G'(a) = \phi'(\ee^{\frac{a}{s}}) \cdot \frac{1}{s} \ee^{\frac{a}{s}}.
\]
Thus $\phi$ is concave if and only if $\phi'$ is a decreasing function, which in turn is equivalent to $g(a) \cdot \ee^{-\frac{a}{s}}$ being decreasing in $a$. Since $g$ is piece-wise continuously differentiable, we know that even if $g$ is not differentiable at some point $a$, the left and right derivatives do exist, and they also satisfy $g'(a)/g(a) \leq 1/s$. From this it follows that $g(a) \cdot \ee^{-a/s}$ is indeed decreasing, concluding the proof of Theorem~\ref{thm:main}.

\bigskip

In fact, the above proof establishes a stronger statement: given any background risk $W$ and any number $s \geq 0$, it is dominant to accept every gamble $X$ with positive expectation and riskiness $\mathrm{R}(X) \leq s$ \emph{if and only if} the background risk satisfies $\mathrm{S}(W) \geq s$. In particular, the bound $\mathrm{R}(X) \leq \mathrm{S}(W)$ in Theorem \ref{thm:main} cannot be improved.

To conclude this section we briefly comment on the case of general background risks that need not have full support or admit a density. Note that \eqref{eq:diffidence} holds if and only if $G(a) = \phi(\ee^{\frac{a}{s}})$ for some concave function $\phi$. Thus, for any $s > 0$, it is dominant to accept every gamble $X$ with $\E{X} > 0$ and $\mathrm{R}(X) \leq s$ only if the background risk distribution $G$ does admit a density $g$, and $g = G'$ almost everywhere. Moreover, in order for $\phi$ to be concave, it must hold that $g(a) \cdot \ee^{-\frac{a}{s}}$ is decreasing in $a$. These conditions are both necessary and sufficient for it to be dominant to accept all such gambles $X$. In particular, it is not necessary for the background risk to be unbounded from above, or for it to have a differentiable or even continuous density.

\section{Discussion}\label{sec:discussion}

 \paragraph{Limited Liability and Small-Stakes Risk Aversion.} As shown by Corollary~\ref{cor:scaled-gamble}, our main result holds under comparatively weaker assumptions in the case of limited liability, where it is not necessary to impose assumptions on the tails of the background risk distribution. This is intuitive, as details of the left tail of $W$ are irrelevant for describing the distributions of the truncated random variables $(W)_\ell$ and $(W + X)_\ell$.
 
 It is useful to compare this fact with a different phenomenon related to limited liability. It is well known that a decision maker, even if risk-averse, can display a behavior that appears to be risk-loving when protected by limited liability. Consider, for instance, an expected utility preference with utility function $u$. Under limited liability, payoffs are truncated below the cutoff $\ell$. This in turn is equivalent to applying a convex transformation to $u$. It is then easy to construct examples where limited liability makes a gamble $X$ that has negative expectation acceptable despite $u$ being globally concave. Thus, a natural question is whether the conclusions of Corollary~\ref{cor:scaled-gamble} are driven by a degree of risk-loving that is a by-product of limited liability, rather than by the effect of background risk.
 
 As we show in \S\ref{sec:limited-liability-risk-aversion} in the Appendix, limited liability does not lead to risk-loving behavior when considering small gambles. We prove that given a gamble $X$ with strictly \emph{negative} expectation, any risk-averse and expected utility decision maker will reject the gamble $tX$ for all $t > 0$ small enough, provided the background risk $W$ is non-trivial. Thus, while limited liability can in some cases make decision makers who are risk-averse appear to be risk-loving, this phenomenon does not occur with small gambles.

 
 \paragraph{Exponential Tails and Constant Absolute Risk Aversion.} As is well known, under CARA expected utility, a decision maker accepts or rejects a gamble independently of wealth levels, and thus also independently of background risk. This apparent contradiction to our main results is explained by the fact that CARA rules out heavy-tailed distributions if we additionally require expected utilities to be finite.\footnote{To be specific, consider a CARA decision maker with risk aversion level $\alpha$, together with a gamble $X$ that she \emph{rejects} without any background risk. Then the riskiness index satisfies $\mathrm{R}(X) \geq \nicefrac{1}{\alpha}$. By Theorem \ref{thm:main}, we can find a background risk $W$ with size $\mathrm{S}(W) \geq \mathrm{R}(X)$, such that $W + X$ first-order stochastically dominates $W$. Nonetheless, since $\mathrm{S}(W) \geq \mathrm{R}(X) \geq \nicefrac{1}{\alpha}$, it can be shown that the CARA expected utilities of $W$ and $W + X$ are both $-\infty$. For example, if $W$ has Laplace distribution with density $g(x) = \frac{\beta}{2} \ee^{-\beta \vert x \vert}$, then $\mathrm{S}(W) \geq \nicefrac{1}{\alpha}$ requires $\beta \leq \alpha$. Denoting by $u(x) = -\ee^{-\alpha x}$ the CARA utility function, it is then easy to see $\int_{-\infty}^{\infty} u(x) \cdot g(x) \,\dd x = -\infty$.}

 Due to the fact that infinite expected utilities cannot be compared, there is no contradiction between the assumption that a CARA decision maker rejects a gamble $X$ and our conclusion that she finds it dominant to accept $X$ under a heavy-tailed background risk $W$.

 The same argument shows that in the absence of limited liability, the assumption that the background risk $W$ is heavy left-tailed is necessary for our results to hold (as already discussed at the end of Section \ref{sec:main-proof}). If $W$ has thin tails, then a sufficiently risk-averse CARA utility $u$ that satisfies $\E{u(X)} < u(0)$ would continue to satisfy $\E{u(X + W)} < \E{u(W)}$, with both expectations finite. Thus thin-tailed background risks cannot make every monotone preference accept the gamble $X$. 
 
 This technical issue notwithstanding, it is worth mentioning that heavy-tailed distributions have a long history in modeling risk and have seen a number of economic applications \citep*[see, e.g.,][]{morris2019crises}.
 

 
 \paragraph{Choice Between Two Gambles.} Large background risk leads to risk-neutral behavior not only in the choice between accepting or rejecting a given small gamble, but also in the choice between two small gambles $X$ and $Y$. As we show in \S\ref{sec:choice-between-two-gambles} in the Appendix,  given two gambles $X$ and $Y$ with $\E{X} > \E{Y}$, for any background risk $W$ with sufficiently heavy tails both on the left and on the right, the resulting distribution of $X + W$ first-order stochastically dominates that of $Y + W$. 
 
 \paragraph{Second-Order Stochastic Dominance.} In \S\ref{sec:second-order-dominance} in the Appendix, we consider decision makers whose preferences are monotone with respect to second-order stochastic dominance. This is a stronger assumption that is natural in the study of risk aversion. In this setting we prove a result that is analogous to our main Theorem \ref{thm:main}: every actuarially favorable gamble is accepted when the left tail of the background risk is heavy enough. The measure of tail-heaviness is different in this case, requiring less background risk than the first-order stochastic dominance case.

\newpage
\appendix
\section{Proofs Omitted from the Main Text}

\begin{proof}[Proof of Corollary~\ref{cor:scaled-gamble}]
Suppose $W$ is heavy left-tailed, i.e.\ it satisfies $\mathrm{S}(W) > 0$. As shown by \cite*{aumann2008economic}, the riskiness index $\mathrm{R}$ is positive homogeneous, i.e.\ it satisfies $\mathrm{R}(tX) = t\mathrm{R}(X)$ for all $t>0$. Thus, the riskiness of $tX$ is lower than the exponential size of $W$ for all $t$ small enough. It then follows from Theorem~\ref{thm:main} that accepting $tX$ is dominant.
  
Now suppose the decision maker is protected by limited liability. We prove the following analogue of Theorem \ref{thm:main}, which will imply this part of Corollary~\ref{cor:scaled-gamble} as well as Corollary~\ref{cor:normal}.

\begin{proposition}\label{prop:limited-liability}
Suppose the decision maker is protected by limited liability bound $\ell$. Then under any background risk $W$, it is dominant to accept every gamble $X$ with positive expectation and riskiness
\[
\mathrm{R}(X) \leq \left(\sup_{a \geq \ell - \max[X]} \frac{g'(a)}{g(a)} \right)^{-1}.
\]
\end{proposition}

Compared to the definition of $\mathrm{S}(W)$, the supremum on the right-hand side above only considers points above $\ell - \max[X]$. We now show that Proposition \ref{prop:limited-liability} implies the limited liability case of Corollary \ref{cor:scaled-gamble}. Since we assumed $g$ to be strictly positive and piece-wise continuously differentiable, the ratio $g'(a)/g(a)$ is bounded on every compact interval. Moreover, as $g$ is eventually decreasing, this ratio is bounded for points $a \geq \ell - 1$. Thus, for $t > 0$ small enough, $\mathrm{R}(tX) = t \cdot \mathrm{R}(X)$ is close to zero, while 
\[
\left(\sup_{a \geq \ell - \max[tX]} \frac{g'(a)}{g(a)} \right)^{-1} \geq \left(\sup_{a \geq \ell - 1} \frac{g'(a)}{g(a)} \right)^{-1} 
\]
is bounded away from zero. Applying Proposition \ref{prop:limited-liability} to the gamble $tX$ yields Corollary \ref{cor:scaled-gamble}. 
\end{proof}
 
\begin{proof}[Proof of Proposition~\ref{prop:limited-liability}]
We follow the proof of Theorem \ref{thm:main} in \S\ref{sec:main-proof}, and explain the necessary modifications. Under limited liability, accepting the gamble $X$ is dominant if
\[
\P{W + X \leq a}\leq \P{W \leq a} = G(a) \, \text{~for all~} a \geq \ell \,.
\]
Thus, instead of \eqref{eq:eu-representation}, we only need to check
\begin{equation}\label{eq:eu-representation-limited-liability}
    \E{ G(a-X) } \leq G(a) \, \text{~for all~} a \geq \ell \,.
\end{equation}
Let $s$ denote $\left(\sup_{a \geq \ell - \max[X]} \frac{g'(a)}{g(a)} \right)^{-1}$, and suppose without loss that $s > 0$. Then $\mathrm{R}(X) \leq s$ implies $\E{\ee^{-X/s}} \leq 1$. So the CARA decision maker with utility function $U(a) = \ee^{a/s}$ would reject the gamble $-X$. 

On the other hand, the definition of $s$ implies that the function $G$ has weakly higher Arrow-Pratt index than $U$ on the interval $[\ell-\max[X], \infty)$, so $G$ is a concave transformation of $U$ \emph{on this interval}. Since the comparison between $a$ and $a-X$ in \eqref{eq:eu-representation-limited-liability} only involves wealth levels that are above $\ell-\max[X]$, we deduce that rejection of $-X$ by the utility function $U$ implies rejection by $G$. Hence the proposition.
\end{proof}
 
\begin{proof}[Proof of Corollary~\ref{cor:normal}]
By Proposition \ref{prop:limited-liability}, accepting $X$ is dominant under limited liability if 
\[
\mathrm{R}(X) \leq \left(\sup_{a \geq \ell - \max[X]} \frac{g'(a)}{g(a)} \right)^{-1},
\]
where $g(a) = \frac{1}{\sqrt{2\pi}\sigma}\ee^{-\frac{(a-\mu)^2}{2\sigma^2}}$ denotes the density of the normal background risk. It is well known that $g$ is log-concave, so that $g'(a)/g(a)$ is decreasing in $a$. Thus 
\[
\sup_{a \geq \ell - \max[X]} \frac{g'(a)}{g(a)} = \frac{g'(\ell-\max[X])}{g(\ell-\max[X])} = \frac{\mu - \ell + \max[X]}{\sigma^2}.
\]
Thus the normal background risk makes $X$ dominant whenever $\mathrm{R}(X) \leq \frac{\sigma^2}{\mu - \ell + \max[X]}$, or equivalently $\sigma \geq \sqrt{\mathrm{R}(X)(\mu - \ell + \max[X])}$. 
\end{proof}

\section{On Limited Liability and Risk Aversion}\label{sec:limited-liability-risk-aversion}

In this section we consider an expected utility decision maker endowed with a utility function $u\colon \RR \to \RR$ that is differentiable, strictly increasing, concave, and bounded from above. The decision maker is protected by limited liability bound $\ell$. We assume that the background risk $W$ satisfies $\P{W = \ell} = 0$ and $\P{W > \ell} > 0$ (these conditions are weaker than the assumption described in \S\ref{sec:model}). We have the following result:

\begin{proposition}\label{prop:limited-liability-risk-loving}
Let $X$ be any bounded gamble with \emph{negative} expectation. Then under the above conditions, it holds that for all $t > 0$ small enough, 
\[
    \E{u((W + tX)_\ell)} < \E{u(W_\ell)}\,.
\]
\end{proposition}

Thus, it is optimal for the decision maker to reject any small gamble $tX$ with negative expectation. This result clarifies that for small gambles, limited liability does not lead to risk-loving behavior (under general assumptions on preferences). The result has an additional implication. In the absence of limited liability, it is known that no background risk can make it dominant to accept a gamble with negative expectation: a necessary condition for $W + X$ to dominate $W$ in first-order stochastic dominance is that $X$ has positive expectation.\footnote{This is clear when $W$ has finite expectation, since then  $\E{X} = \E{W+X}-\E{W} > 0$. The general case was shown by \cite*{ pomatto2018stochastic}.} The result here shows that the same conclusion holds under limited liability, for small gambles.

 
 
\begin{proof}[Proof of Proposition \ref{prop:limited-liability-risk-loving}]
Without loss of generality, we set the liability bound to be $\ell = 0$. Let $W$ satisfy $\P{W = 0} = 0$ and $\P{W > 0} > 0$. Under our assumptions on $u$, the expectations $\E{u( (W+tX)^+ )}$ and $\E{u(W^+ )}$ are finite, where we use $z^{+}$ to denote $(z)_0 = \max\{z, 0\}$. Normalizing $u(0) = 0$, we can compute the payoff difference between accepting and rejecting $tX$ as follows:
\begin{align*}
    \lefteqn{\E{u( (W+tX)^+ )-u(W^+ )}}\\ 
    &= \E{\mathbf{1}_{W+tX > 0} \cdot u( W+tX )-\mathbf{1}_{W > 0} \cdot u(W)} \\
    &=\E{\mathbf{1}_{W > 0}\cdot [u( W+tX )-  u(W)] } + \E{ [\mathbf{1}_{W+tX > 0} - \mathbf{1}_{W > 0}] \cdot u( W+tX )}\\
    &\leq\E{\mathbf{1}_{W > 0} \cdot [u( W+tX )-  u(W)]} + \E{ \vert \mathbf{1}_{W+tX > 0} - \mathbf{1}_{W > 0}\vert \cdot \vert u(W + tX) \vert} \,.
\end{align*}
By the concavity of $u$, we have
\[
    u( W + t X )-  u(W) \leq u'(W) \cdot tX \,,
\]
where we interpret the derivative $u'$ as the left-derivative in this proof. Thus, we get that
\[
    \E{\mathbf{1}_{W > 0} \cdot [u( W+tX )-  u(W)]} \leq \E{\mathbf{1}_{W > 0}u'(W) tX} = t \cdot \E{\mathbf{1}_{W > 0}u'(W)} \cdot \E{X} \,.
\]
Using the assumption that $W > 0$ with positive probability, $u' > 0$ (which follows from strict monotonicity and concavity) and $\E{X} < 0$, we obtain
\[
    \limsup_{t \to 0} \frac{1}{t} \E{\mathbf{1}_{W > 0} \cdot [u( W+tX )-  u(W)]} \leq \E{\mathbf{1}_{W > 0}u'(W)} \cdot \E{X} < 0 \,.
\]

Next, suppose $\vert X \vert \leq L$. Then $\mathbf{1}_{W+tX > 0} \neq \mathbf{1}_{W > 0}$ if and only if $W + tX > 0 \geq W$ or $W + tX \leq 0 < W$. In both cases we obtain $\vert W \vert \leq tL$ and $\vert W + tX \vert \leq tL$. Thus 
\[
    \E{ \vert \mathbf{1}_{W+tX > 0} - \mathbf{1}_{W > 0}\vert \cdot \vert u(W + tX) \vert} \leq \P{\vert W \vert \leq tL} \cdot \sup_{\vert a \vert \leq tL} \vert u(a) \vert.
\]
Taking the limit as $t \to 0$ yields
\[
    \limsup_{t\to 0} \frac{1}{t}\E{ \vert \mathbf{1}_{W+tX > 0} - \mathbf{1}_{W > 0}\vert \cdot \vert u(W + tX) \vert} \leq \limsup_{t\to 0 } \P{\vert W \vert \leq tL} \cdot \frac{1}{t} \sup_{\vert a \vert \leq tL} \vert u(a) \vert = 0\,,
\]
where we used the fact that $\lim_{t\to 0 } \P{\vert W \vert \leq tL} = \P{W = 0} = 0$ and $\lim_{t \to 0} \frac{1}{t} \sup_{\vert a \vert \leq tL} \vert u(a) \vert = \lim_{t \to 0} \frac{1}{t} \cdot -u(-tL) = L\cdot u'(0) < \infty$.

Combining the above estimates, we conclude that
\[
\limsup_{t \to 0} \frac{1}{t} \E{u( (W+tX)^+ )-u(W^+ )} < 0.
\]
Hence for sufficiently small $t$, the decision maker would reject the scaled gamble $tX$ under limited liability.
\end{proof}

We conclude with two observations. We assumed that $W$ takes values above $\ell$ with positive probability. Clearly, if $W \leq \ell$ with certainty, then it is dominant to accept every gamble under limited liability. It is also important for the result that $\E{X}$ is strictly negative. If the gamble $X$ has zero expectation, then it can be dominant to accept $tX$ for every small $t$, so long as the background risk distribution $G$ is concave on the interval $(\ell - \epsilon, \infty)$ for some positive $\epsilon$ (since this guarantees \eqref{eq:eu-representation-limited-liability} for $tX$).

\section{Second-Order Dominance}\label{sec:second-order-dominance}

Our analysis can be extended to the smaller class of risk-averse preferences. We say that accepting $X$ is \textit{dominant for a risk-averse decision maker} if $W + X$ dominates $W$ with respect to second-order stochastic dominance. We also introduce a modified version of the exponential size: for any background risk $W$ with c.d.f.\ G, let
\[
    \mathrm{S}_2(W) =\left(\sup_{a\in \mathbb{R}}\frac{g(a)}{G(a)}\right)^{-1}.
\]
It is easy to show that $\mathrm{S}_2(W) \geq \mathrm{S}(W)$.\footnote{If $\mathrm{S}(W) = 0$ then the result is trivial. If instead $\mathrm{S}(W) > 0$, then we have the inequality $g(x) \geq g(y) \cdot \ee^{\frac{x-y}{\mathrm{S}(W)}}$. Note that $G(y) = \int_{-\infty}^y g(x) \, \dd x \to 0$ as $y \to -\infty$. Using the previous inequality, we deduce that $g(y) \to 0$ as $y \to -\infty$. Hence, for each $a$, it holds that 
\[
\frac{g(a)}{G(a)} = \frac{\int_{-\infty}^{a} g'(x) \, \dd x}{\int_{-\infty}^{a} g(x) \, \dd x} \leq \sup_{x} \frac{g'(x)}{g(x)} = \frac{1}{\mathrm{S}(W)}.
\]
As a result, $\mathrm{S}_2(W) \geq \mathrm{S}(W)$ again holds.} 

\begin{theorem}\label{thm:second-order}
Under any given background risk $W$ with finite expectation, it is dominant for a risk-averse decision maker to accept every gamble $X$ with positive expectation and riskiness $\mathrm{R}(X) \leq \mathrm{S}_2(W)$.
\end{theorem}

\begin{proof}[Proof of Theorem~\ref{thm:second-order}]
Let $s = \mathrm{S}_2(W)$ and without loss focus on $s > 0$. By a well-known characterization of second-order stochastic dominance, it is dominant to accept $X$ if and only if
\begin{equation}\label{eq:proof-second-order}
    \int_{-\infty}^a \P{W+X \leq t} \,\dd t \leq \int_{-\infty}^a \P{W \leq t} \,\dd t \, \text{~for all~} a\in\RR.
\end{equation}
That the integrals in \eqref{eq:proof-second-order} are finite follows from the fact that $W$ and $W + X$ have finite expectations. By Tonelli's Theorem, the quantity $\int_{-\infty}^a \P{W+X \leq t} \,\dd t$ is equal to
\[
    \int_{-\infty}^a \E{G(t-X)} \,\dd t = \E{\int_{-\infty}^a G(t-X)\,\dd t} = \E{\int_{-\infty}^{a-X} G(t)\,\dd t}.
\]
Hence, it is second-order dominant to accept a gamble $X$ if and only if for every $a \in \RR$
\[
    \E{u_G(a-X)} \leq u_G(a),
\]
where $u_G(a) = \int_{-\infty}^{a} G(t)\,\dd t$. Therefore, as in the proof of Theorem~\ref{thm:main}, we obtain that accepting $X$ is dominant if 
\begin{equation}\label{eq:diffidence-ssd}
    \E{\ee^{-\frac{1}{s} X }} \leq 1 \implies \E{ u_G(a - X) } \leq u_G(a) \, \text{~for all~} a\in\RR.
\end{equation}
Equation \eqref{eq:diffidence-ssd} holds whenever $u_G$ is globally more risk-averse than the CARA utility function $U(a) = \ee^{\frac{a}{s}}$. The Arrow-Pratt index for $u_G$ is $-g(a)/G(a)$, which by assumption is weakly larger than $-1/s$, the Arrow-Pratt index for $U$. Thus $u_G$ is indeed more risk-averse than $U$, concluding the proof.
\end{proof}

\section{Choice Between Two Gambles}\label{sec:choice-between-two-gambles}

In this section, we extend the analysis to situations where the decision maker faces a choice between two bounded gambles $X$ and $Y$ that have distinct distributions $F_X$ and $F_Y$. We say it is \emph{dominant to choose $X$ over $Y$} under background risk $W$, if $W + X$ first-order stochastically dominates $W + Y$. A result similar to Theorem~\ref{thm:main} can be obtained if we consider background risks with heavy tails both on the left and on the right. For this we define the \textit{two-sided exponential size}
\[
\mathrm{S}^*(W) = \left(\sup_{a} \Bigg| \frac{g'(a)}{g(a)} \Bigg| \right)^{-1},
\]
which is equal to $\min\{ \mathrm{S}(W), \mathrm{S}(-W) \}$. Then we have:
\begin{theorem}\label{thm:two-gambles-both-tails}
The following are equivalent:
\begin{itemize}
    \item[(i)] $\mathbb{E}[X] > \mathbb{E}[Y]$; 
    \item[(ii)] there exists $s \in (0, \infty)$ such that under any background risk $W$ with $\mathrm{S}^*(W) \geq s$, choosing $X$ over $Y$ is dominant. 
\end{itemize}
\end{theorem}

\begin{proof}[Proof of Theorem \ref{thm:two-gambles-both-tails}]
We first show (ii) implies (i). Given any finite $s$, we can choose $W$ to have a Laplace distribution with sufficiently large variance. Then $W$ satisfies $\mathrm{S}^*(W) \geq s$, and by assumption $W + X$ must first-order stochastically dominate $W + Y$. Since such a $W$ has finite expectation, we have $\mathbb{E}[W+X] \geq \mathbb{E}[W+Y]$, which implies $\mathbb{E}[X] \geq \mathbb{E}[Y]$. The inequality is in fact strict, for otherwise $W + X$ would have the same distribution as $W + Y$, and $X$ would have the same distribution as $Y$.\footnote{The last claim can be proved by considering the moment generating function in a neighborhood of $0$. Since $\mathbb{E}[\ee^{tW}]$ is finite for $t$ close to $0$, both $\mathbb{E}[\ee^{t(W+X)}]$ and $\mathbb{E}[\ee^{t(W+Y)}]$ are finite and are equal. It follows that $\mathbb{E}[\ee^{tX}] = \mathbb{E}[\ee^{tY}]$ for $t$ in a neighborhood of $0$, which implies $X$ and $Y$ have the same distribution.}

To prove (i) implies (ii), we assume $\mathbb{E}[X] > \mathbb{E}[Y]$ and take $s$ to be a large positive number (to be determined later). Consider any background risk $W$ with $\mathrm{S}^*(W) \geq s$, i.e.\ the density $g$ satisfies $\vert g'(a) / g(a) \vert \leq 1/s$ for all $a$. Let $h(a) = \ln g(a)$, then we can rewrite the condition as 
\[
\vert h'(a) \vert \leq \frac{1}{s} \, \text{~for all~} a \in \mathbb{R}.
\]
We now use this to show $\P{W + Y \leq a} \geq \P{W + X \leq a}$ for all $a$. Since $W$ is independent from both $X$ and $Y$, this comparison is equivalent to 
\[
\int_{-M}^{M} g(a - z) \cdot F_Y(z) \, \dd z \geq \int_{-M}^{M} g(a - z) \cdot F_X(z) \, \dd z,
\]
where $M$ is a large number such that $[-M, M]$ contains the support of both $X$ and $Y$. This in turn is equivalent to
\[
\int_{-M}^{M} \ee^{h(a - z)} \cdot (F_Y(z) - F_X(z))\, \dd z \geq 0.
\]
Dividing both sides by $\ee^{h(a)}$, we just need to show that for all $a$ 
\[
\int_{-M}^{M} \ee^{h(a - z) - h(a)} \cdot (F_Y(z) - F_X(z))\, \dd z \geq 0.
\]
Observe that since $\vert h' \vert$ is bounded above by $1/s$, we have $\vert h(a - z) - h(a) \vert \leq M/s$ for all $a \in \mathbb{R}$ and all $z \in [-M, M]$. Thus if $s$ is chosen to be sufficiently large, then the above integral converges, uniformly across $a$, to the integral $\int_{-M}^{M} (F_Y(z) - F_X(z))\, \dd z$. Since this limit integral evaluates to $\mathbb{E}[X] - \mathbb{E}[Y] > 0$, the result follows. 
\end{proof}

If we only know that the background risk has a heavy left tail (as in Theorem \ref{thm:main}), then the condition $\E{X} > \E{Y}$ is no longer sufficient to guarantee the dominance of $X$. Below we derive the suitable condition in this case. We say that \emph{$X$ strongly dominates $Y$ in the convex order}, if $\max[X] > \max[Y]$ and 
\begin{equation}\label{eq:convex-order}
    \int_{a}^{\infty} (F_Y(z) - F_X(z)) \, \dd z > 0 \, \text{~for all~} a < \max[X].
\end{equation}
In particular, this requires $\E{X} > \E{Y}$ in the limit $a \to -\infty$. 

To interpret this condition, note that $X$ dominates $Y$ in the convex order if and only if $-Y$ dominates $-X$ in second-order stochastic dominance. In other terms, $X$ can be obtained from $Y$ by a combination of mean-preserving spreads and right-ward mass shifts. Conversely, if $X$ is obtained from $Y$ by replacing \emph{each} realization $y$ of $Y$ by a gamble with expectation \emph{strictly greater} than $y$, then $X$ \emph{strongly} dominates $Y$ in the convex order. This is a natural generalization of the case studied in the main text, where $Y$ is a constant and $X$ is any gamble with a higher expectation. 

\begin{theorem}\label{thm:two-gambles}
Suppose $\max[X] \neq \max[Y]$. Then the following are equivalent:
\begin{itemize}
    \item[(i)] $X$ strongly dominates $Y$ in the convex order;
    \item[(ii)] there exists $s \in (0, \infty)$ such that under any background risk $W$ with $\mathrm{S}(W) \geq s$, choosing $X$ over $Y$ is dominant. 
\end{itemize}
\end{theorem}

\begin{proof}[Proof of Theorem \ref{thm:two-gambles}]
As in the proof of Theorem \ref{thm:main}, choosing $X$ over $Y$ is dominant if and only if 
\[
\mathbb{E}[G(a-X)] \leq \mathbb{E}[G(a-Y)] \, \text{~for all~} a\in\RR \,.
\]
Since we want this to hold for all background risks $G$ with exponential size $\geq s$, and since the exponential size is translation-invariant, it is without loss to restrict to the case of $a = 0$. That is, we seek to understand the conditions under which 
\[
\mathbb{E}[G(-X)] \leq \mathbb{E}[G(-Y)] \, \text{~for all~} G \text{~with exponential size~} \geq s.
\]
As before, let $U(a) = \ee^{\frac{a}{s}}$ denote a risk-loving CARA utility function. Then $G$ has exponential size at least $s$ if and only if $G(a) = \phi(U(a))$ for some increasing concave function $\phi$.\footnote{To be fully rigorous, we also need $g(a) = \phi'(\ee^{a/s}) \cdot \frac{1}{s} \ee^{a/s}$ to be strictly positive, continuously differentiable, and eventually decreasing. These additional restrictions on $\phi$ do not affect the subsequent analysis because on any compact domain, any increasing concave function can be uniformly approximated by another increasing concave function with these additional properties.} Thus, the above comparison can be rewritten as 
\[
\mathbb{E}\left[\phi\left(\ee^{\frac{-X}{s}}\right)\right] \leq \mathbb{E}\left[\phi\left(\ee^{\frac{-Y}{s}}\right)\right] \, \text{~for all increasing concave functions~} \phi.
\]
In other terms, the random variable $\tilde{Y} = \ee^{\frac{-Y}{s}}$ should dominate $\tilde{X} = \ee^{\frac{-X}{s}}$ with respect to second-order stochastic dominance. 

Let $\tilde{F}_X$ and $\tilde{F}_Y$ denote the c.d.f.\ of $\tilde{X}$ and $\tilde{Y}$, respectively. Then second-order stochastic dominance holds if and only if (noting that $\tilde{X}$ and $\tilde{Y}$ are both supported on $\mathbb{R}_{+}$):
\[
\int_{0}^{c} (\tilde{F}_X(t) - \tilde{F}_Y(t)) \, \dd t \geq 0 \, \text{~for all~} c > 0.
\]
If we write $t = \ee^{-\frac{z}{s}}$, then $\tilde{F}_X(t) = 1 - F_X(z)$, $\tilde{F}_Y(t) = 1 - F_Y(z)$. Changing variables in the above integral, and denoting $a=-s\ln(c)$, we obtain the following equivalent condition (modulo a factor of $1/s$): 
\begin{equation}\label{eq:convex-order-with-exponential}
\int_{a}^{\infty} (F_Y(z) - F_X(z)) \cdot \ee^{-\frac{z}{s}} \, \dd z \geq 0 \, \text{~for all~} a \in \mathbb{R}.
\end{equation}
Below we show that when the maxima of $X$ and $Y$ are different, the above condition holds for some positive $s$ if and only if $X$ strongly dominates $Y$ in the convex order.

In one direction, suppose $\max[X] > \max[Y]$ and \eqref{eq:convex-order} holds. Then intuitively \eqref{eq:convex-order-with-exponential} would also hold if $s$ is large, in which case the integrand $(F_Y(z) - F_X(z)) \cdot \ee^{-\frac{z}{s}}$ is close to $F_Y(z) - F_X(z)$. This can be formalized by observing that we only need to prove \eqref{eq:convex-order-with-exponential} for $a$ in the compact interval $\min[X] \leq a \leq \max[Y]$. As $s \to \infty$ the integral $\int_{a}^{\infty} (F_Y(z) - F_X(z)) \cdot \ee^{-\frac{z}{s}} \, \dd z$ converges uniformly to $\int_{a}^{\infty} (F_Y(z) - F_X(z)) \, \dd z$ on this interval. Since this limit is a continuous function in $a$ and strictly positive on this interval, it is bounded away from $0$. Thus by uniform convergence, there exists some large $s$ such that \eqref{eq:convex-order-with-exponential} holds. 
 
For the converse, suppose \eqref{eq:convex-order-with-exponential} holds for some $s$. Then there cannot exist some $a$ with $F_Y(a) < 1 = F_X(a)$, since otherwise \eqref{eq:convex-order-with-exponential} fails at this point $a$. It follows that $\max[X] \geq \max[Y]$, and the inequality is in fact strict by the assumption that $\max[X] \neq \max[Y]$. As a result, $F_Y(z) - F_X(z)$ is strictly positive for $z \in [\max[Y], ~\max[X])$, and \eqref{eq:convex-order-with-exponential} holds with strict inequality for $a$ in the same interval. We now use this to prove \eqref{eq:convex-order}. Observe that
\begin{align*}
\lefteqn{\int_{a}^{\infty} (F_Y(z) - F_X(z))  \, \dd z }\\ &=\ee^{\frac{a}{s}} \int_{a}^{\infty} (F_Y(z) - F_X(z)) \cdot \ee^{-\frac{z}{s}} \, \dd z + \int_{a}^{\infty} \left(\frac{\ee^{\frac{c}{s}}}{s} \cdot  \int_{c}^{\infty} (F_Y(z) - F_X(z)) \cdot \ee^{-\frac{z}{s}} \, \dd z \right)\, \dd c.
\end{align*}
So from \eqref{eq:convex-order-with-exponential}, we must have $\int_{a}^{\infty} (F_Y(z) - F_X(z))  \, \dd z \geq 0$. Moreover, the inequality is strict because in the double integral on the RHS above, the term $\int_{c}^{\infty} (F_Y(z) - F_X(z)) \cdot \ee^{-\frac{z}{s}} \, \dd z$ is strictly positive for any $c \in [\max[Y], ~\max[X])$. For any $a < \max[X]$, the mass of such $c > a$ is strictly positive. Hence \eqref{eq:convex-order} holds with strict inequality, completing the proof.
\end{proof}

\section{Additional Results}

\begin{proposition}\label{prop:risk-bound}
For any gamble $X$ that is supported on $[-M, M]$ and has expectation $\epsilon > 0$, its riskiness index satisfies $\mathrm{R}(X) \leq \frac{M^2}{\epsilon}$.
\end{proposition}

\begin{proof}[Proof of Proposition~\ref{prop:risk-bound}]
Let $\lambda = \frac{\epsilon}{M^2}$. We first show that $\mathbb{E}[\ee^{-\lambda X}] \leq 1$. Indeed, since $\epsilon = \mathbb{E}[X] \leq M$, we have $\lambda \leq \frac{1}{M}$. As $X \in [-M, M]$ with probability one, we have $-\lambda X \in [-1, 1]$. In this range, it always holds that $\ee^{-\lambda X} \leq 1 - \lambda X + (\lambda X)^2$. Hence $\mathbb{E}[\ee^{-\lambda X}] \leq 1 - \lambda \mathbb{E}[X] + \lambda^2 \mathbb{E}[X^2] \leq 1 - \lambda \epsilon + \lambda^2 M^2 = 1$.
 
Now consider the function $f(a) = \mathbb{E}[\ee^{-aX}]$, defined for $a \geq 0$. It is easy to see that $f(0) = 0$ and $f$ is strictly convex. Thus, $\frac{1}{\mathrm{R}(X)}$ is the unique number $c > 0$ such that $f(c) = 0$. Since we just proved that $f(\lambda) \leq 0$, convexity implies $c \geq \lambda$. In other words $\frac{1}{\mathrm{R}(X)} \geq \frac{\epsilon}{M^2}$. 
\end{proof}

\section{Background Risk for Cumulative Prospect Theory Preferences}\label{appx:CPT}

Table~\ref{tab:CPT-table} below shows the levels of background risk needed to make a decision maker with cumulative prospect theory (CPT) preferences to accept various gambles. The specific CPT preference we consider has gain/loss probability weighting functions $w^+(p)=\frac{p^\gamma}{(p^\gamma+(1-p)^\gamma)^{\nicefrac{1}{\gamma}}}$, $w^-(p)=\frac{p^\delta}{(p^\delta+(1-p)^\delta)^{\nicefrac{1}{\delta}}}$ with $\gamma=0.61,\delta=0.69$, loss aversion parameter $\lambda=2.25$ and value function $v(x)= x^{0.88}$ for $x \geq 0$ and $v(x) = -\lambda(-x)^{0.88}$ for $x < 0$. These parameter values are taken from \citet[pages 309--312]{tversky1992advances}.

 \bgroup
\def\arraystretch{1.4}
\begin{table}[t]
    \centering
    \begin{tabular}{l|l l l}
        \toprule
         \textbf{Gamble} & \multicolumn{3}{c}{\textbf{StDeviation of Background Risk: $\sigma$}} \\
         Gain/Loss & Laplace &  Logistic & Normal \\
         \midrule
         $\$11/\$10$ & $\sigma \geq \$62$ & $\sigma \geq \$46$ & $\sigma \geq \$44$ \\
         $\$55/\$50$ & $\sigma \geq \$306$ & $\sigma \geq \$230$ & $\sigma \geq \$217$ \\
         $\$110/\$100$ & $\sigma \geq \$612$ & $\sigma \geq \$460$ & $\sigma \geq \$434$ \\
         $\$550/\$500$ & $\sigma \geq \$3058$ & $\sigma \geq \$2299$ & $\sigma \geq \$2169$ \\
         $\$1000/\$1100$ & $\sigma \geq \$6115$ & $\sigma \geq \$4598$ & $\sigma \geq \$4338$ \\
         \bottomrule
    \end{tabular}
    \caption{Standard deviation of background risk sufficient for a CPT decision maker to accept various fifty-fifty gambles under different distributional assumptions on the background risk.  \label{tab:CPT-table}}
\end{table}
\egroup

\bigskip

\bibliography{refs}

\end{document}